\documentclass{amsart}
\usepackage{amsmath}
\usepackage{graphicx}
\usepackage{amssymb}
\usepackage{bbm}
\usepackage{mathrsfs}
\usepackage{amssymb}
\usepackage{amsthm}
\usepackage{latexsym}
\usepackage{indentfirst}
\usepackage{enumitem}
\usepackage{anysize}
  \usepackage{graphics} %% add this and next lines if pictures should be in esp format
  \usepackage{epsfig} %For pictures: screened artwork should be set up with an 85 or 100 line screen
\usepackage{graphicx}  \usepackage{epstopdf}%This is to transfer .eps figure to .pdf figure; please compile your paper using PDFLeTex or PDFTeXify.
 \usepackage[colorlinks=true]{hyperref}
\hypersetup{urlcolor=blue, citecolor=red}

\newcommand{\F}{\mathbb{F}}

\newtheorem{theorem}{Theorem}[section]

\newtheorem{lemma}[theorem]{Lemma}
\newtheorem{proposition}[theorem]{Proposition}

\theoremstyle{definition}

\begin{document}

%% Place the running title of the paper with 40 letters or less in []
 %% and the full title of the paper in { }.
\title{Two-weight codes from trace codes over $R_k$}
      \author[MinJia Shi, Yue Guan]{}
 \email{smjwcl.good@163.com}
\email{guanyueeee@163.com}
 \thanks{The first author is supported by NNSF of China (61202068),
Technology Foundation for Selected Overseas Chinese Scholar, Ministry of Personnel of China (05015133) and the Open Research Fund of National Mobile Communications Research Laboratory, Southeast University (2015D11). }

% Add corresponding author at the footnote of the first page if it is necessary.
% Plase add $^*$ adjacent to the corresponding author's name on the first page.
% The example shown in this template is if the first author is the corresponding author.
\thanks{$^{**}$ Corresponding author: }

\maketitle
 % Do not forget to end the {\footnotesize by the sign }

\medskip \centerline{\scshape MinJia Shi$^{**}$}
\medskip
{\footnotesize
% please put the address of the first author
   \centerline{School of Mathematical Sciences,}
   \centerline{Anhui University,}
   \centerline{ Hefei, Anhui Province 230039, PR China}
  }
\medskip \centerline{\scshape Yue Guan$^*$}
{\footnotesize
\centerline{School of Mathematical Sciences,}
   \centerline{Anhui University,}
   \centerline{ Hefei, Anhui Province 230039, PR China}}

% Do not forget to end the {\footnotesize by the sign }
\bigskip

% The name of the associate editor will be entered by an editorial staff
% "Communicated by the associate editor name" is not needed for special issue.
 %\centerline{(Communicated by the associate editor name)}

%The abstract of your paper
\begin{abstract}
We construct a family of two-Lee-weight codes over the ring $R_k,$ which is defined as trace codes with algebraic structure of abelian codes. The Lee weight distribution of the two-weight codes is given. Taking the Gray map, we obtain optimal abelian binary two-weight codes by using the Griesmer bound. An application to secret sharing schemes is also given.
\end{abstract}

\section{Introduction}
For the reason that the access structures of the secret sharing schemes derived from few weights linear codes can be completely determined \cite{DD, YD}, linear codes with few weights are of great interest in secret sharing schemes. Two-Lee-weight codes over fields have been studied since 1970s due to their connections to strongly regular graphs, difference sets and finite geomertry \cite{CK, D}. Recently, the notion of trace codes are extended from fields to rings and some trace codes over rings are introduced. In \cite{SL1, SL2}, two infinite families of two-Lee-weight codes over the rings $\F_2+u\F_2$ and $\F_2+u\F_2+v\F_2+uv\F_2$ are constructed. On the basis of the study of the ring $R_k$ in \cite{DY3}, we shall construct a family of two-Lee-weight codes over $R_k$, which is a generalization of the rings $\F_2+u\F_2$ and $\F_2+u\F_2+v\F_2+uv\F_2.$

In this paper, we construct a family of two-weight codes which are provably abelian but perhaps not cyclic. These codes are constructed as trace codes. We determine the weight distributions of these codes by exponential character sums. Taking the Gray map, we obtain a family of binary abelian codes and they are shown to be optimal by the application of the Griesmer bound. Moreover, an application to secret sharing schemes is also given.

The paper is organized as follows. In the next section, we give some definitions and propositions about the ring $R_k$, the Gray map and the Lee weight that we need in the following parts of the paper. Section 3 shows some information of the codes over the ring $R_k.$ Section 4 shows that the trace codes are abelian. Section 5 computes the weight distribution of our codes and proves the optimality of them by the application of the Griesmer bound.
Section 6 determines the minimum distance of the dual codes, proves that all these codes are minimal and then describes an application to secret sharing schemes. In section 7, we make a summary of this paper and give some interesting problems.

\section{Definition}
\subsection{Rings}
Consider the ring $R_k=\mathbb{F}_2[u_1,u_2,...,u_k]/\langle u_i^2=0,u_iu_j=u_ju_i\rangle$ with $k \geq 1.$ The ring can also be defined recursively as $$R_k=R_{k-1}[u_k]/\langle u_k^2=0,u_ku_j=u_ju_k\rangle, j=1,2,...,k-1.$$

In order to represent the elements of $R_k$ conveniently, we define $u_A:=\prod\limits _{i\in A}u_i$ for any subset $A\subseteq\{1,2,...,k\}.$ What's more, $u_A=1$ when $A=\emptyset$ by the convention. Then any element of $R_k$ can be represented as $\sum\limits _{A\subseteq\{1,2,...,k\}}c_Au_A, c_A\in \mathbb{F}_2.$ The ring $R_k$ can be extended to $\mathcal{R}=\mathbb{F}_{2^m}[u_1,u_2,...,u_k]/\langle u_i^2=0,u_iu_j=u_ju_i\rangle$ with a given positive integer $m.$ The elements of $\mathcal{R}$ are in the form of $\sum\limits _{A\subseteq\{1,2,...,k\}}c_Au_A, c_A\in \mathbb{F}_{2^m}.$ An element of $\mathcal{R}$ is a unit if the coefficient of $c_\emptyset\neq 0$ and it is a maximal ideal if $c_\emptyset=0.$ Let $\mathcal{R}^*$ be a set that contains all the units of $\mathcal{R}$ and $M$ denote maximal ideal. Hence, $\mathcal{R}=\mathcal{R}^*\bigcup M.$

There is a Frobenius operator $F$ which maps $\sum\limits _{A\subseteq\{1,2,...,k\}}c_Au_A$ to $\sum\limits _{A\subseteq\{1,2,...,k\}}c_A^2u_A.$ Let $Tr$ denote the \emph{Trace function}, which is defined as $Tr=\sum_{j=0}^{m-1}F^j.$ It is obvious that $Tr\big(\sum\limits _{A\subseteq\{1,2,...,k\}}c_Au_A\big)=\sum\limits _{A\subseteq\{1,2,...,k\}}u_Atr(c_A)$ and $tr()$ denotes the standard trace from $\mathbb{F}_{2^m}$ to $\mathbb{F}_2$.
\subsection{The Lee weight and the Gray map}
The Lee weight was defined in \cite{YK} as the Hamming weight of the image of the codeword under the Gray map $\phi_k$ and \cite{DY3} has given a recursive definition of $\phi_k.$ For $\bar{c}\in R_k^n,$ it can be represented as $\bar{c}_1+u_k\bar{c}_2$ with $\bar{c}_1,\bar{c}_2\in R_{k-1}^n.$ Then, the Gray map is defined as follows: $$\phi_k(\bar{c})=(\phi_{k-1}(\bar{c}_2),\phi_{k-1}(\bar{c}_1)+\phi_{k-1}(\bar{c}_2)).$$ With the definition of the Lee weight, we obtain a distance preserving map from $R_k^n$ to $\mathbb{F}_2^{2^kn}.$
\section{Codes}
A linear code over $R$ of length $n$ is a $R$-submodule of $R^n.$ Assume that $\sum_A c_Au_A$ and $\sum_B d_Bu_B$ with $A,B\subseteq\{1,2,...,k\}$ are two elements of the ring $R$, then their standard inner product $(\sum_A c_Au_A)(\sum_B d_Bu_B)=\sum\limits _{A,B\subseteq\{1,2,...,k\}, A\cap B=\emptyset}c_Ad_Bu_{A\cup B}$ which is defined over $R.$ Denote the dual code of $C$ by $C^\perp$ and define it as $$C^\perp =\big\{\sum_B d_Bu_B\big |\big(\sum_A c_Au_A\big)\big(\sum_B d_Bu_B\big)=0, \forall\sum_A c_Au_A\in C\big\}.$$
\section{Symmetry}
In this paper, we define $L=\mathcal{R}^*$, $n=|L|=(2^m-1)2^{m(2^k-1)}$ and $N=2^kn.$ For $a\in \mathcal{R},$ we define the vector $ev(a)$ by the evaluation map $ev(a)=(Tr(ax))_{x\in L}.$ The code $C$ is defined by the formula $C=\{ev(a)|a\in \mathcal{R}\}.$ Thus, the length of the code $C$ is $n$ while the length of the code $\phi_k (C)$ is $N.$
\begin{proposition}
The group $L$ acts regularly on the coordinates of $C.$
\end{proposition}
\begin{proof}
For any $u,v\in L,$ the change of variables $x\mapsto (u/v)x$ maps $v$ to $u.$ With given $u,v,$ a permutation with this property is unique.
\end{proof}
The code $C$ is thus an abelian code based on the group $L.$ In other words, it is an ideal of the group ring $R[\mathcal{R}^*]$. $L$ is not a cyclic group, thus $C$ is not likely to be  cyclic.

%%%%%%%%%%%%%%%%%%%%%%%%%%%%%%%%%%%%%%%%%%%%%%%%%%%%%%%%%%%%%%%%%%%%%%%%%%%%%%%%%%%%%%%%%%%%%%%%%%%%%%%%%%%%%%%%%%%¨´

\section{Weight distribution}
Before calculating the Lee weight of the codewords of the code $C,$ we recall some classic lemmas first.
\begin{lemma}(\cite{MS}, $(6)$, p. $412$)\label{a}
If $y=(y_1,y_2,...,y_n)\in \mathbb{F}_2^n,$ then $2w_H(y)=n-\sum\limits_{i=1}^n(-1)^{y_i}.$
\end{lemma}
\begin{lemma}(\cite{MS}, Lemma 9, p. $143$)
For any $z\in \mathbb{F}_{2^m}^*,$ $\sum\limits _{x\in \mathbb{F}_{2^m}}(-1)^{tr(zx)}=0$ is set up. Apparently, $\sum\limits _{x\in \mathbb{F}_{2^m}^*}(-1)^{tr(zx)}=-1.$
\end{lemma}
\begin{theorem}
For $a\in \mathcal{R},$ the Lee weight of the codewords of $C$ is as follows:
\begin{itemize}
  \item  If $a=0,$ then $w_L(ev(a))=0;$
  \item  If $a\in M\setminus\{0\}$, then\\
            if $a=c_{\{1,...,k\}}u_{\{1,...,k\}}$ and $c_{\{1,...,k\}}\in \mathbb{F}_{2^m}^*$, then $w_L(ev(a))=2^{2^km+k-1},$\\
            if $a\in M\setminus\{0,c_{\{1,...,k\}}u_{\{1,...,k\}}\},$ then $w_L(ev(a))=2^{k-1}(2^m-1)2^{m(2^k-1)};$
  \item  If $a\in \mathcal{R}^*$, then $w_L(ev(a))=2^{k-1}(2^m-1)2^{m(2^k-1)}.$
\end{itemize}
\end{theorem}
\begin{proof}
Divide the value of $a$ into three categories to discuss the Lee weight of the codewords.
\begin{itemize}
  \item  When $a=0,$ it is obvious that $w_L(ev(a))=0.$
  \item  If $a\in M\setminus\{0\}$ and $a=c_{\{1,...,k\}}u_{\{1,...,k\}}$ with $c_{\{1,...,k\}}\in \mathbb{F}_{2^m}^*$, let $x=\sum\limits_{B\subseteq\{1,...,k\}} d_Bu_B\in L.$
On the basis of the definition of the inner product, we compute $ax$ and here are the process and result:
\begin{equation*}
\begin{aligned}
ax&=\sum\limits_{B\subseteq\{1,...,k\}} c_{\{1,...,k\}}u_{\{1,...,k\}}d_Bu_B\\
&=c_{\{1,...,k\}}d_\emptyset u_{\{1,...k\}}, \ {\rm where} \ d_\emptyset\in \mathbb{F}_{2^m}^*.
\end{aligned}
\end{equation*}
Then $Tr(ax)=tr(c_{\{1,...,k\}}d_\emptyset)u_{\{1,...,k\}}.$
Taking the Gray map, we get $\phi_k(ev(a))=(\underbrace{tr(c_{\{1,...,k\}}d_\emptyset),...,tr(c_{\{1,...,k\}}d_\emptyset)}_{2^k}).$

According to Lemma 5.1, we have $$2^k|L|-2w_L(ev(a))=2^k(\sum\limits_{d_\emptyset\in \mathbb{F}_{2^m}^*}\sum\limits_{d_B\in \mathbb{F}_{2^m}}(-1)^{tr(c_{\{1,...,k\}}d_\emptyset)}),$$ where $B\neq\emptyset\subseteq\{1,...,k\}$ and the right side of the formula can be written as $-2^k2^{m(2^k-1)}.$ Then, $w_L(ev(a))=2^{2^km+k-1}.$

In addition, if $a\in M\setminus\{0,c_{\{1,...,k\}}u_{\{1,...,k\}}\},$ suppose $a=c_{\{1\}}u_{\{1\}}, c_{\{1\}}\in \mathbb{F}_{2^m}^*,$ and $x=\sum\limits_{B\subseteq\{1,...,k\}}d_Bu_B\in L.$
Next, determine the value of $ax$ in the same way as the previous case. So $ax=\sum\limits_{B\subseteq\{2,...,k\}}c_{\{1\}}d_Bu_{\{1\}\cup B}$ and $Tr(ax)=\sum\limits_{B\subseteq\{2,...,k\}}tr(c_{\{1\}}d_B)u_{\{1\}\cup B}.$ Taking the Gray map, every element of the vector $\phi_k(ev(a))$ contains $tr(c_{\{1\}}d_{\{2,...,k\}})$ which implies that $2^k|L|-2w_L(ev(a))=0$ since $\sum\limits_{d_{\{2,...,k\}}\in \mathbb{F}_{2^m}}(-1)^{tr(c_{\{1\}}d_{\{2,...,k\}})}=0.$
For other $a\in M\setminus\{0,u_{\{1\}},u_{\{1,...,k\}}\},$ the Lee weight equals $2^{k-1}|L|$ as well.
  \item If $a\in\mathcal{R}^*$, let $a=\sum\limits _{A\subseteq\{1,2,...,k\}}c_Au_A$, where $c_{A\setminus\{0\}}\in \mathbb{F}_{2^m}$ and $c_\emptyset\in \mathbb{F}_{2^m}^*.$ Moreover, $x=\sum\limits _{B\subseteq\{1,2,...,k\}}d_Bu_B$, where $d_{B\setminus\{0\}}\in \mathbb{F}_{2^m}$ and $d_\emptyset\in \mathbb{F}_{2^m}^*.$ Computing $ax$ and $Tr(ax)$, we discover that every element of the vector $\phi_k(ev(a))$ contains $tr(c_{\emptyset}d_{\{1,...,k\}})$ which implies $2^k|L|-2w_L(ev(a))=0$ since $\sum\limits_{d_{\{1,...,k\}}\in \mathbb{F}_{2^m}}(-1)^{tr(c_{\emptyset}d_{\{1,...,k\}})}=0.$ Hence, $w_L(ev(a))=2^{k-1}|L|.$
\end{itemize}
\end{proof}
Thus, we have constructed a family of two-weight codes of length $2^k(2^m-1)2^{m(2^k-1)}$ and dimension $2^km.$ Their weight distributions and  respective frequencies are shown in Table I below.
\begin{center}$\mathrm{Table~ I. }~~~\mathrm{weight~ distribution~ of}~ \phi_k(C) $\\
\begin{tabular}{cccc||cc}
\hline
  Weight&&   & & Frequency  \\
  \hline

  0        & &   & & 1\\
  $w_1=2^{k-1}2^{m(2^k-1)}(2^m-1)$        & &   &              &$2^{m2^k}-2^m$\\
  $w_2=2^{k-1}2^{m(2^k-1)}2^m$  &    & &       &$2^m-1$ \\
  \hline
\end{tabular}
\end{center}
Next, we study the optimality of $\phi_k(C)$ by the application of the Griesmer bound. Recall the Griesmer bound applied to an $[n,K,d]$ binary code first. For an $[n,K,d]$ code, we have $\sum_{i=0}^{K-1}\lceil\frac{d}{2^i}\rceil\leq n.$ If the code with parameters $[n,K,d+1]$ does not exist, then the $[n,K,d]$ code is optimal.
\begin{theorem}
For arbitrary integers $m\geq 2$ and $k\geq 1,$  the code $\phi_k (C)$ is optimal.
\end{theorem}
\begin{proof}
The parameters of $\phi_k (C)$ are $[N,K,d]=[2^k(2^m-1)2^{m(2^k-1)},2^km,2^{k-1}(2^m-1)2^{m(2^k-1)}].$ We claim that $\sum_{j=0}^{K-1}\lceil \frac{d+1}{2^j}\rceil >N.$ There are two different values of $\sum_{j=0}^{k-1}\lceil \frac{d+1}{2^j}\rceil$ depending on the value of $j$ and we discuss these two cases next.
\begin{itemize}
 \item $0\leq j\leq m(2^k-1)+k-1 \Rightarrow \lceil \frac{d+1}{2^j} \rceil =2^{m(2^k-1)+k-1-j}(2^m-1)+1$;
 \item $m(2^k-1)+k-1< j\leq2^km-1 \Rightarrow \lceil \frac{d+1}{2^j} \rceil =2^{m2^k+k-1-j}$.
\end{itemize}
Then,
\begin{equation*}
\begin{aligned}
\sum_{j=0}^{K-1}\big\lceil \frac{d+1}{2^j}\big\rceil -N &=\sum_{j=0}^{m(2^k-1)+k-1}\big\lceil\frac{d+1}{2^j}\big\rceil+\sum_{j=m(2^k-1)+k}^{m2^k-1}\big\lceil\frac{d+1}{2^j}\big\rceil-N\\
&=(2^k-1)(m-1)+k
\end{aligned}
\end{equation*}
In order to prove $\sum_{j=0}^{K-1}\lceil \frac{d+1}{2^j}\rceil -N>0,$ we can prove $(2^k-1)(m-1)+k>0$ equivalently. Discussing the range of $m$ and $k,$ we obtain that the inequation is true when $k\geq 1$ and $m\geq 2.$
\end{proof}

\section{Application to secret sharing scheme}
The access structure of secret sharing scheme constructed from linear code is hard to determine. However, if all the codewords of the linear code are minimal, we can use its dual code to construct an interesting secret sharing scheme. So next, we check that whether the codewords of $\phi_k(C)$ are minimal or not. The basic property minimal vector of a given binary linear code is described by the following lemma \cite{AB}.
\subsection{Minimal codewords}
\begin{lemma}(Ashikhmin-Barg) Denote by $w_0$ and $w_{\infty}$ the smallest and largest nonzero weights of a binary code $C$, respectively. If
$$\frac{w_0}{w_{\infty}}>\frac{1}{2},$$ then every nonzero codeword of $C$ is minimal.
\end{lemma}
\begin{theorem}
 All the nonzero codewords of $\phi_k(C),$ for $m\geq 2$ are minimal.
\end{theorem}

\begin{proof}
 By the preceding lemma with $w_0=w_1,$ and $w_{\infty}=w_2.$ Rewriting the inequality of Lemma 6.1 as $2w_1-w_2>0,$ we end up with the condition
$2^{k-1}2^{m(2^k-1)}(2^m-2)>0,$ which is satisfied for $m\geq 2$. Hence, the theorem is proved.
\end{proof}
\subsection{The dual code}
\begin{lemma}
For all $a\in \mathcal{R},$ if $Tr(ax)=0,$ then $x=0.$
\end{lemma}
\begin{proof}
Assume that $a=\sum_Ac_Au_A$ and $x=\sum_Bd_Bu_B$ where $A, B\subseteq\{1,...,k\}$ and $c_A, d_B\in \mathbb{F}_{2^m}.$ Thus $ax=\sum\limits _{A,B\subseteq\{1,2,...,k\}, A\cap B=\emptyset}c_Ad_Bu_{A\cup B}.$ Then $Tr(ax)=0$ means $tr(c_Ad_B)=0$ where $A, B\subseteq\{1,...,k\}$ and $A\cap B=\emptyset.$ Using the nondegenerate character of $tr()$ \cite{MS}, we obtain that $d_B=0$ for all $B\subseteq\{1,...,k\}$ which is equivalent to $x=0.$
\end{proof}
Next, we determine the dual Lee distance of the two-Lee-weight code $C.$
\begin{theorem}
For all positive integers $k$ and $m\geq2,$ the dual Lee distance $d'$ of $C$ is $2$.
\end{theorem}
\begin{proof}
Here we use the proof by contradiction. Assume that $d'\geq3$ to prove that $d'<3.$ With the sphere-packing bound, we have $2^{2^km}\geq1+N.$ Rewriting the inequation, we have $2^k<2^{k-m}+1$ which is impossible for any positive integers $k$ or $m\geq2.$ Hence, $d'<3.$

Then, we prove that $d'=2.$ If not, $C^\perp$ has at least one codeword with Lee weight one and assume that it has value $\gamma$ at some $x\in\mathcal{R}^*.$ Since $\mathcal{R}^*$ is the set that contains all the units of $\mathcal{R}$, for all $a\in\mathcal{R},$ $\gamma Tr(ax)=0$ implies $Tr(ax\gamma)=0.$ By Lemma 6.3, we know that $x=0$, which is contradict with $x\in\mathcal{R}^*.$
In conclusion, the dual Lee distance $d'$ is equal to $2.$
\end{proof}
\subsection{Massey's scheme}
Massey's scheme  \cite{YD}, is a secret sharing scheme which based on Coding theory. When all nonzero codewords are minimal, it was shown in \cite{DY2} that there is the following alternative, depending on $d'.$
\begin{itemize}
 \item If $d'\ge 3,$ then the SSS is \emph{``democratic''}: every user belongs to the same number of coalitions.
 \item If $d'=2,$  then there are users  who belong to every coalition: the \emph{``dictators''.}
\end{itemize}
By Theorems 6.2 and 6.4, we see that the Secret Sharing Scheme built on $\phi_k(C)$ is dictatorial.
\section{Conclusion}
In the present paper, we have studied an infinite family of trace codes. Because the localizing set $L$ has the structure of an abelian multiplicative group, the trace code
is an abelian code, that is an ideal in the group ring of $L.$ It is not clear whether the code is cyclic or not. In fact, we also study the trace code over the localizing set $\{\sum_Ca_Cu_C:C\subseteq\{1,2,...,k\},a_{C\setminus\{\emptyset\}}\in\mathbb{F}_{2^m},a_\emptyset\in \mathbf{D}\}$ \cite{HY1, HY2} and here $D$ is a subset of $\mathbb{F}_{2^m}^*.$ Two families of codes are obtained. One family of codes is the same as that constructed with the localizing set $L$. Another family of codes is a three-weight codes and it can be used to construct secret sharing schemes. More importantly, it is worthwhile to study other trace codes over the extensions of $R_k.$
%\section{Acknowledgement}
%The authors are grateful to the reviewers and the Editor for their helpful comments that improved the presentation and quality of this paper. This research is
%supported by National Natural Science Foundation of China (61202068), the Open Research Fund of National Mobile Communications Research Laboratory, Southeast University (2015D11),
%Technology Foundation for Selected Overseas Chinese Scholar, Ministry of Personnel of China (05015133) and
%Key Projects of Support Program for outstanding young talents in Colleges and Universities (gxyqZD2016008).
%%%%%%%%%%%%%%%%%%%%%%%%%%%%%%%%%%%%%%%%%%%%%%%%%%%%

\end{document}